\documentclass[copyright,creativecommons]{eptcs}
\usepackage{breakurl}             
\usepackage{gastex}
\usepackage{framed,amsfonts}
\usepackage{color}
\usepackage{url}
\usepackage{graphicx}
\usepackage{rotating}
\usepackage{amsmath}
\newtheorem{lemma}{Lemma}
\newtheorem{theorem}{Theorem}
\newtheorem{definition}{Definition}
\newcommand{\eop}{\hspace*{\fill}$\Box$}
\newenvironment{proof}{{\it Proof.}\quad}{\eop\vspace*{4mm}}
\def\To{\Rightarrow}

\def\rank{{\sf rank}}
\def\val{{\sf val}}
\title{XPath Node Selection over Grammar-Compressed Trees}
\author{Sebastian Maneth
\institute{School of Informatics\\University of Edinburgh\\UK}
\email{smaneth@inf.ed.ac.uk}
\and
Tom Sebastian
\institute{Innovimax and\\Links Project (INRIA and LIFL, Lille)\\FR}
\email{tom.sebastian@inria.fr}
}
\def\titlerunning{XPath Node Selection over Grammar-Compressed Trees}
\def\authorrunning{S. Maneth, T. Sebastian}
\begin{document}
\maketitle
\begin{abstract}
XML document markup is highly repetitive and therefore well compressible
using grammar-based compression. 
Downward, navigational XPath can be executed over grammar-compressed 
trees in PTIME: the query is translated into an automaton which is 
executed in one pass over the grammar.
This result is well-known and has been mentioned before.
Here we present precise bounds on the time complexity
of this problem, in terms of big-O notation.
For a given grammar and XPath query, we
consider three different tasks:
(1)~to count the number of nodes selected by the query
(2)~to materialize the pre-order numbers of the selected nodes, and
(3)~to serialize the subtrees at the selected nodes.
\end{abstract}

\section{Introduction}

An XML document represents the serialization of an ordered
node-labeled unranked tree.
These trees are typically highly repetitive with respect to their
internal node labels. This was observed by Buneman, Koch, and 
Grohe when they showed that the minimal DAGs of such trees (where
text and attribute values are removed) have only 10\% of the number
of edges of the trees~\cite{DBLP:conf/vldb/KochBG03}. 
The DAG removes repeating subtrees and
represents each distinct subtree only once. A nice feature of such a 
``factorization'' of repeated substructures, is that many queries can
be evaluated directly on the compressed factored representation, without
prior decompression~\cite{DBLP:conf/vldb/KochBG03,DBLP:conf/lics/FrickGK03}. 
The sharing of repeated subtrees can be
generalized to the sharing of repeated (connected) subgraphs of the tree,
for instance using the sharing graphs of Lamping~\cite{DBLP:conf/popl/Lamping90}, 
or the
straight-line (linear) context-free tree (SLT) grammars of Busatto, Lohrey,
and Maneth~\cite{DBLP:journals/is/BusattoLM08}. 
The recent ``TreeRePair'' compressor~\cite{lohmanmen13}
shrinks the (edge) size of typical XML document trees by a factor of four,
with respect to the minimal unranked DAG (cf. Table~4 in~\cite{lohmanmen13}).

It was shown by Lohrey and Maneth~\cite{DBLP:journals/tcs/LohreyM06} 
that tree automata and 
navigational XPath can be evaluated in PTIME over SLT grammars, without prior
decompression. 
This is used to build a system for selectivity
estimation for XPath by Fisher and Maneth~\cite{DBLP:conf/icde/FisherM07}.
Roughly speaking, the idea is to translate the XPath query into a 
certain tree automaton, and to execute this automaton over the SLT grammar.
In this paper we make these constructions more precise and give 
complexity bounds in terms of big-O notation.
We use the ``selecting tree automata'' of Maneth and Nguyen~\cite{DBLP:journals/pvldb/ManethN10}
(see also~\cite{DBLP:conf/icde/ArroyueloCMMNNSV10}), in their deterministic variant.
Similar variants of selecting tree automata have been
considered in~\cite{DBLP:conf/fsttcs/NeumannS98,DBLP:journals/tcs/NevenS02,DBLP:conf/dbpl/NiehrenPTT05}.
We explain how XPath queries containing the child, descendant, and following-sibling
axes can be translated into our selecting tree automata.
It is achieved via a well-known translation of such XPath queries
into DFAs, due to Green, Gupta, Miklau, Onizuka,
and Suciu~\cite{DBLP:journals/tods/GreenGMOS04}.
We then study three different tasks: (1) to count the number of nodes that
a deterministic top-down selecting tree automaton selects on a tree
represented by a given SLT grammar, 
(2) to materialize the pre-order numbers of the selected nodes, and 
(3) to serialize, in XML syntax, the depth-first left-to-right traversal
of the subtrees rooted at the selected nodes.
The first problem can be solved in $O(|Q||G|)$ where $Q$ is the state set
of the automaton, and $G$ the SLT grammar. 
The second and third problem can be solved in time $O(|Q||G|+r)$ 
and time $O(|Q||G|+s)$, respectively, 
where $r$ is the number of selected nodes and $s$ is the length of
the serialization of the selected subtrees.
Note that the length $s$ can be quadratic in the size of the tree
represented by $G$ (e.g., if every node is selected). 
Thus, $s$ is of length $(2^{|G|})^2$ if $G$ compresses exponentially. 
We show how to obtain a compressed representation of this
serialization by a straight-line string grammar $G'$ 
of size $O(|Q||G|r)$.

Most of the constructions of this paper are
implemented in the ``TinyT'' system. TinyT
and a detailed experimental evaluation is given by
Maneth and Sebastian~\cite{DBLP:journals/corr/abs-1012-5696}.

\section{Preliminaries}

\textbf{XML trees.}\quad
XML defines several different node types, such as 
element, text, attribute, etc.
Here we are only concerned with element nodes.
Our techniques can easily be applied to other types of nodes. 
An \emph{unranked XML tree} 
is a finite node-labeled ordered unranked tree.
The node labels are non-empty words over a fixed finite alphabet 
$\mathbf{U}$. The \emph{first-child next-sibling encoding} of 
an unranked XML tree $t$ is the binary tree obtained from $t$ as follows:
if a node in $t$ has a first child, then this first child becomes the left
child in the binary tree. If a node has a next sibling in $t$, then this
next sibling becomes the right child in the binary tree. If 
a node has no first child (resp. next sibling) then in the binary tree
its left (resp. right) child is a leaf labeled with the special label $\_$.
There is a one-to-one correspondence between
unranked XML trees and binary trees with internal nodes labeled in 
$\mathbf{U}^+$ and leaves labeled $\_$ (and with a root whose right child
 is a $\_$-leaf).
We only deal with such binary trees from now on,
and refer to them as \emph{XML trees}.
Figure~\ref{fig:lib} shows an unranked XML tree on the left (albeit a binary tree itself), 
and is first-child next-sibling encoded tree on the right.

\smallskip

\textbf{Tree grammars.}\quad
A ranked set consists of a set $A$ together with a mapping
$\rank$ which associates the non-negative integer $\rank(a)$
to each $a\in A$.
We fix a special set of symbols $Y=\{y_1, y_2, \dots\}$ 
called \emph{parameters}.

A \emph{straight-line (linear) tree grammar} (for short, \emph{SLT grammar})
is a tuple $G=(N,S,P)$ where $N$ is a finite ranked set of nonterminals, 
$S\in N$ is the start nonterminal of rank $0$, 
and $P$ maps each $A\in N$ of rank $k$ to an ordered finite tree $t$.
In $t$ there is exactly one leaf labeled $y_i$,
for each $1\leq i\leq k$, and the $y_1,\dots,y_k$ appear
in pre-order of $t$.
Nodes in $t$ are labeled by nonterminals in $N$, by 
words in $\mathbf{U}^+$, and by the special leaf symbol $\_$.
If a node is labeled by a nonterminal of rank $k$, then it has 
exactly $k$ children. If a node is labeled by a word in $\mathbf{U}^+$
then it has exactly two children.
If $P(A)=t$ then we also write $A\to t$ 
or $A(y_1,\dots,y_k)\rightarrow t$ and refer to 
this assignment as  a ``production'' or a ``rule''.
We require that the relation $H_G$, called the 
\emph{hierarchical order of $G$}, and defined as
\[
H_G=
\{(A,B)\in N\times N\mid B\text{ occurs in }P(A)\}
\] 
is acyclic and connected.
The grammar $G$ produces exactly one tree, denoted by 
$\val(G)$. It can be obtained by repeatedly
replacing nonterminals $A\in N$ by their definition $P(A)$, starting
with the initial tree $P(S)$. Replacements are done in the obvious
way: a subtree $A(t_1,\dots,t_k)$ is replaced by the tree $P(A)$ in which
$y_i$ is replaced by $t_i$ for $1\leq i\leq k$.
We define the \emph{rank} of the grammar as the maximum of the ranks
of all its nonterminals.
We extend the mapping val to nonterminals $A$ and define $\val(A)$ 
as the tree obtained from $A(y_1,\dots,y_k)$ by applying the rules of $G$
(and treating the $y_i$ as terminal symbols). The tree $\val(A)$
is a binary tree with internal nodes in $\mathbf{U}^+$ and leaves labeled
$\_$ or $y_i$. Each $y_i$ with $1\leq i\leq k$ occurs once, and 
$y_1,\dots,y_k$ occur in pre-order of $\val(A)$.

The \emph{size of an SLT grammar $G$} is defined as the sum of sizes of
the right-hand side trees of all rules. The \emph{size of a tree} is defined
as its number of edges.

\smallskip

\textbf{Example.} Consider the SLT grammar $G_1$ with
three nonterminals $S$, $B$, and $T$, of ranks zero, one, and
zero, respectively. It consists of the following productions:
\[
\begin{array}{lcl}
S      &\rightarrow&\text{lib}(B(B(\_)),\_)\\
B(y_1) &\rightarrow&\text{book}(T,y_1)\\
T      &\rightarrow&\text{title}(\_,\text{author}(\_,\_))
\end{array}
\]
It should be clear that the tree $\val(G_1)$ produced
by this grammar is the 
binary tree shown on the right of Figure~\ref{fig:lib}.
\begin{figure}
\noindent
\centering
\gasset{Nframe=n,Nadjust=wh}
\begin{picture}(79,31)(6,0)
  \unitlength=1mm
  \node(root)(21,25){library}
  \node(2)(12,16){book}
  \node(4)(30,16){book}
  \node(5)(7,7){title}
  \node(6)(17,7){author}
  \node(7)(25,7){title}
  \node(8)(35,7){author}
  \drawedge[AHnb=0](root,2){}
  \drawedge[AHnb=0](root,4){} 
  \drawedge[AHnb=0](2,5){} 
  \drawedge[AHnb=0](2,6){} 
  \drawedge[AHnb=0](4,7){}
  \drawedge[AHnb=0](4,8){}

  \node(a)(50,30){library}
  \node(b)(50,24){book}
  \node(c)(50,18){title}
  \node(d)(50,12){\_}
  \node(e)(70,24){\_}
  \node(f)(70,18){book}
  \node(g)(60,12){author}
  \node(h)(60,6){\_}
  \node(i)(65,6){\_}
  \node(j)(70,12){title}
  \node(k)(70,6){\_}
  \node(l)(80,6){author}
  \node(m)(80,0){\_}
  \node(n)(85,0){\_}
  \node(o)(80,12){\_}
  \drawedge[AHnb=0](a,b){}
  \drawedge[AHnb=0](b,c){} 
  \drawedge[AHnb=0](c,d){} 
  \drawedge[AHnb=0](a,e){} 
  \drawedge[AHnb=0](c,g){} 
  \drawedge[AHnb=0](b,f){} 
  \drawedge[AHnb=0](f,j){} 
  \drawedge[AHnb=0](j,k){} 
  \drawedge[AHnb=0](j,l){}  
  \drawedge[AHnb=0](g,h){}  
  \drawedge[AHnb=0](g,i){}  
  \drawedge[AHnb=0](l,m){}  
  \drawedge[AHnb=0](l,n){}  
  \drawedge[AHnb=0](f,o){}

\end{picture}
\caption{\label{fig:lib} Unranked XML tree and its encoded binary tree.}
\end{figure}
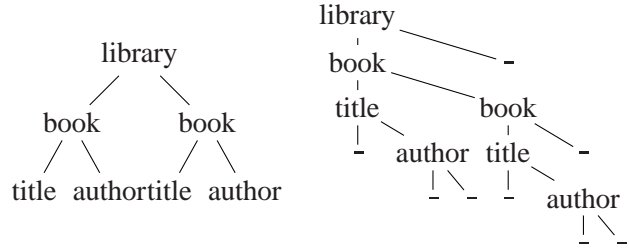

\section{XPath to Automata}
\label{sect:XPath}

We consider XPath queries \emph{without filters}.
In Section~\ref{sect:filters} we explain how filters can be supported.
Such queries are of the form
\[
Q = /a_1::t_1/a_2::t_2/ \cdots /a_n::t_n
\]
where 
$a_i\in\{\text{child}, \text{descendant}, \text{following-sibling}\}$ and
$t_i\in\{*\}\cup\mathbf{U}^+$.
Thus, we support two types of node tests
(i) a local (element) name and
(ii) the wildcard ``*'', and  support three axes:
child, descendant, and following-sibling.
For a query $Q$ and XML tree $t$ we denote by $Q(t)$ the set of
nodes that $Q$ selects on $t$. We do not define this set formally here.

It was shown by Green, Gupta, Miklau, Onizuka,
and Suciu~\cite{DBLP:journals/tods/GreenGMOS04} 
that any XPath query $Q$ containing only the child
and descendant axes can be translated into a deterministic finite state automaton
$\text{DFA}(Q)$. 
Note that their queries and automata also allow to compare
text and attribute values against constants.
The DFA constructed for a given query, is evaluated over the paths of the
unranked XML input tree. When a final state is reached at a node,
then this node is selected by the query.
Their translation is a straightforward extension
of the ``KMP-automata'' for string matching, explained for instance
in the chapter on string matching in~\cite{DBLP:books/mg/CormenLRS01}.
If there are no wildcards in the query, then
Green et al show that the size of the obtained DFA is linear
in the size of the query.
In the presence of wildcards, the DFA size is exponential in the maximal
number of *'s between any two descendant steps
(see Theorem~4.1 of~\cite{DBLP:journals/tods/GreenGMOS04}).
To understand their translation, consider the following example
query:
\[
Q_1=//a/\! *\!/b//c/d
\]
where ``$//$'' denotes the descendant axis (more precisely, it denotes
the query string ``$/\text{descendant}::$''), 
and ``$/$'' denotes the child axis.
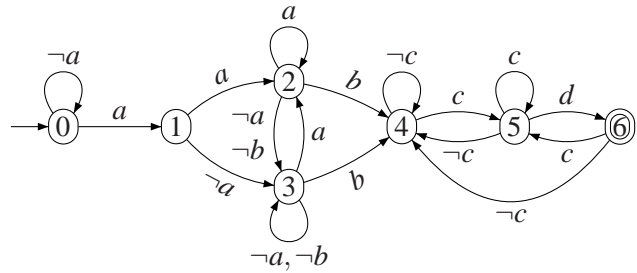
\begin{figure}
\centering
\gasset{Nframe=y,Nadjust=wh}

\begin{picture}(78,33)(0,0)
  \unitlength=1mm

  \node[Nmarks=i](0)(5,18){$0$}
  \node(1)(20,18){$1$}
  \node(2)(35,24){$2$}
  \node(3)(35,10){$3$}
  \node(4)(50,18){$4$}
  \node(5)(65,18){$5$}
  \node[Nmarks=r](6)(79,18){$6$}
  \drawedge(0,1){$a$}
  \drawedge[curvedepth=2,ELpos=50,ELdist=0.5](1,2){{\begin{turn}{20}$a$\end{turn}}}
  \drawedge[curvedepth=-2,ELpos=50,ELdist=-0.5,ELside=r](2,3){$\begin{array}{c}\neg a\\\neg b\end{array}$}
  \drawedge[curvedepth=-2,ELpos=50,ELside=r,ELdist=-0.5](1,3){{\begin{turn}{-28}$\neg a$\end{turn}}}
  \drawedge[curvedepth=-2,ELpos=50,ELside=r](3,2){$a$}
  \drawedge[curvedepth=1,ELpos=50,ELdist=0.5](2,4){{\begin{turn}{-20}$b$\end{turn}}}
  \drawedge[curvedepth=-1,ELpos=50,ELside=r,ELdist=0.2](3,4){{\begin{turn}{28}$b$\end{turn}}}
  \drawedge[curvedepth=2,ELpos=50](4,5){$c$}
  \drawedge[curvedepth=2,ELpos=50](5,4){$\neg c$}
  \drawedge[curvedepth=2,ELpos=50](5,6){$d$}
  \drawedge[curvedepth=2,ELpos=50](6,5){$c$}
  \drawedge[curvedepth=10,ELpos=50](6,4){$\neg c$}

  \drawloop[loopdiam=5,loopangle=90](0){$\neg a$}
  \drawloop[loopdiam=5,loopangle=90](2){$a$}
  \drawloop[loopdiam=5,loopangle=-90,ELdist=0.5](3){$\neg a,\neg b$}
  \drawloop[loopdiam=5,loopangle=90](4){$\neg c$}
  \drawloop[loopdiam=5,loopangle=90](5){$c$}

\end{picture}
\caption{DFA for the XPath query $Q_1=//a/\! *\!/b//c/d$.\label{fig:DFA}.}
\end{figure}
The corresponding automaton $\text{DFA}(Q_1)$ is shown in Figure~\ref{fig:DFA}.
For a sequence of children steps, 
the idea is similar to KMP~\cite{DBLP:journals/siamcomp/KnuthMP77}: 
when reading a new symbol that fails, we compute
the longest current postfix (including the failed symbol; this
is the difference to KMP) that is a prefix of the query string and add
a transition to the corresponding state.
Care has to be taken for wildcards, because then (in general) 
we need to remember the symbol read; in the example (at state $1$) 
it suffices to know whether it is an $a$, or not.

\subsection*{Selecting Tree Automata}
\label{sec:DST}

\emph{Selecting tree automata} are like ordinary top-down
tree automata operating over binary trees. 
They use special ``selecting transitions'' to indicate that
the current node should be selected.
In this paper we use deterministic selecting automata.
Similar such nondeterministic automata have been considered 
by Maneth and Nguyen~\cite{DBLP:journals/pvldb/ManethN10}.
Since the XML trees may contain arbitrary labels in $\mathbf{U}^+$,
we require that each state of the automaton
has one \emph{default rule}. The default rule is applied if
no other rule is applicable.

\begin{definition}\label{def:DST}
A \emph{deterministic selecting top-down tree (DST) automaton} is a triple
$\mathcal{A} = (Q, q_0, R)$ where $Q$ is a finite set of states,
$q_0\in Q$ is the initial state, and
$R$ is a finite set of rules.
Each rule is of one of these forms:
\[
\begin{array}{lcl}
(q,w)&\to&(q_1,q_2)\\
(q,w)&\To&(q_1,q_2)
\end{array}
\]
where $q,q_1,q_2\in Q$ and
$w\in\{\%\}\cup\mathbf{U}^+$. The symbol 
$\%$ is a special symbol not in $\mathbf{U}$. 
Let $q\in Q$. We require that
(1) there is exactly one rule in $R$ with left-hand side
$(q,\%)$, called the \emph{default rule of $q$}, and
(2) for any $w\in\mathbf{U}^+$ there is at most one rule in
$R$ with left-hand side $(q,w)$.
\end{definition}

A rule of the first form is called \emph{non-selecting rule} and
of the second \emph{selecting rule}.
The semantics of a DST automaton should be clear.
It starts reading a tree $t$ in its initial state $q_0$ at the root 
node of $t$. In state $q$ at a node $u$ of $t$ labeled $w\in\mathbf{U}^+$
it moves to the left child into state $q_1$ and to the right child into
state $q_2$, if there is a rule 
$(q,w)\beta(q_1,q_2)$ with $\beta\in\{\to,\To\}$.
If $\beta=\To$, i.e., the rule is selecting, then $u$ is a result node.
If $\mathcal{A}$ has no such rule, then the default rule is applied
(in the same way).
The unique run of $\mathcal{A}$ on the tree $t$ 
determines the set $\mathcal{A}(t)$ of result nodes.

Assume we are given an XPath query $Q$ with child and descendant axes
only and consider its translated automaton $\text{DFA}(Q)$.
It is straightforward to translate the DFA into a DST automaton.
If the DFA moves from $q$ to $q'$ upon reading the symbol $a$,
then the DST automaton has the transition $(q,a)\to(q',q)$; this is because the right child
corresponds to the next sibling of the unranked XML tree, and at that node we should
still remain in state $q$ and not proceed to $q'$.
The DST automaton that corresponds to the DFA of Figure~\ref{fig:DFA}
is:
\[
\begin{array}{lcllcl}
(q_0,a)&\to&(q_1,q_0)\quad & (q_4,c)&\to&(q_5,q_4)\\
(q_0,\%)&\to&(q_0,q_0) & (q_4,\%)&\to&(q_4,q_4)\\
(q_1,a)&\to&(q_2,q_1) & (q_5,c)&\to&(q_5,q_5)\\
(q_1,\%)&\to&(q_3,q_1) & (q_5,d)&\To&(q_6,q_5)\\
(q_2,a)&\to&(q_2,q_2) & (q_5,\%)&\to&(q_4,q_5)\\
(q_2,b)&\to&(q_4,q_2) & (q_6,c)&\to&(q_5,q_6)\\
(q_2,\%)&\to&(q_3,q_2) & (q_6,\%)&\to&(q_4,q_6)\\
(q_3,a)&\to&(q_2,q_3)\\
(q_3,b)&\to&(q_4,q_3)\\
(q_3,\%)&\to&(q_3,q_3)\\
\end{array}
\]

Consider now a general XPath query in our fragment, i.e., one that contains
child, descendant, \emph{and} the following-sibling axes.
Consider each maximal sequence of following-sibling steps.
We can transform it to a DFA by simply treating them as descendant steps and
running the translation of Green et al.
The obtained DFA is transformed into a DST automaton by simply carrying
out the recursion on the second child only, i.e., if the DFA moves
form $q$ to $q'$ on input symbol $a$, then the DST automaton has the 
transition $(q,a)\to(\text{dead},q')$, where ``dead'' is a sink state.
We merge the resulting automata in the obvious way to obtain one final
DST automaton for the query. E.g. for XPath query
\[
/a/\text{following-sibling}::b/c
\]
we obtain the following DST automaton:
\[
\begin{array}{lcl}
(0,a)     &\to& (\text{dead},1)\\
(0,\%)    &\to& (\text{dead},0)\\
(1,b)     &\to& (2,1)\\
(1,\%)    &\to& (\text{dead},1)\\
(2,c)    & \To & (\text{dead},2)\\ 
(2,\%)    &\to& (\text{dead},2)\\ 
(\text{dead},\%) &\to& (\text{dead},\text{dead})\\
\end{array}
\]

\begin{theorem}
For an XPath query $Q$ we can construct a DST automaton $\mathcal{A}$ such
that $\mathcal{A}(t)=Q(t)$ for every tree $t$.
The size of $\mathcal{A}$ can be bounded according to Theorem~4.1 of~\cite{DBLP:journals/tods/GreenGMOS04}.
In particular, if there are no wildcards, then the size $|\mathcal{A}|$ of $\mathcal{A}$ is in 
$O(|Q|)$.
\end{theorem}

\section{Automata over SLT Grammars}

This section describes how to perform counting, materialization, and 
serialization for the set of nodes $\mathcal{A}(t)$ selected by the
DST automaton $\mathcal{A}$ on the tree $t=\val(G)$ given by
the SLT grammar $G$. Note that the case of counting was already
described by Fisher and Maneth~\cite{DBLP:conf/icde/FisherM07}; they consider
queries with filters and containing more axes than in our 
fragment (e.g., supporting the following axes), and therefore obtain 
higher complexity bounds (cf. Section~\ref{sect:filters}).

\subsection{Counting}
We build a ``count evaluator'' which executes
in one pass over the grammar, counting the
number of result nodes of the given XPath query. The idea is to
memoize the ``state-behavior'' of each nonterminal of the SLT grammar,
plus the number of nodes it selects.

\begin{theorem}
\label{theorem:count}
Given an SLT grammar $G$ and a DST automaton $\mathcal{A}=(Q,q_0,R)$, 
and assuming that operations on integers of size
$\leq |\val(G)|$ can be carried out in constant time,
we can compute the number $|\mathcal{A}(\val(G))|$ of
nodes selected by $\mathcal{A}$ on $\val(G)$ in time
$O(|Q||G|)$.
\end{theorem}
\begin{proof}
Let $G=(N,S,P)$ and let $H_G$ be its hierarchical order
$H_G=\{(A,B)\in N\times N\mid B\text{ occurs in }P(A)\}$.
We compute a mapping $\varphi$ in one pass
 through the rules of $G$, in reverse order
of $H_G$, i.e., starting with those nonterminals $A$ for which
$P(A)$ does not contain nonterminals.
For each nonterminal $A$ of rank $k$ and state $q\in Q$ we define
$\varphi(A,q)=(q_1,\dots,q_k,n)$ where $q_i\in Q$ and
$n$ is a non-negative integer. The $q_i$ are chosen in such a way
that if we run $\mathcal{A}$ on $P(A)$ then we reach the $y_i$-leaf
in state $q_i$, and $n$ is the number of selected nodes of this run.
We start in state $q$ at the root node of $P(A)$, and set our
result counter for this run to zero.
If we meet a nonterminal $B$ during this run, say, in state $q'$,
then its $\varphi$ value is already defined; thus,
$\varphi(B,q')=(q_1',\dots,q_m',n')$. We continue the run at state
$q_i'$ at the $i$-th child of this nonterminal in $P(A)$.
We also increase our result counter for $q$ and $A$ by $n'$.
If we meet a selected terminal node, then we increase the result
counter by one. The final result count is stored as the number $n$
in the last component of the tuple in $\varphi(A,q)$.
Finally, when we are at the start nonterminal $S$, we compute
its entry $\varphi(S,q_0)=(n)$. This number $n$ is the desired
value $|\mathcal{A}(G)|$. Since we process $|Q|$-times each node of a right-hand
side of the rules of $G$, we obtain the stated time complexity.
\end{proof}

\subsection{Materializing}

Here we want to produce an ordered list of pre-order numbers of those
nodes that are selected by a given DST automaton over
an SLT grammar $G$.
Clearly, this cannot be done in time $O(|Q||G|)$ 
because the list can be of length $|\val(G)|$.

First we produce a new SLT grammar $G'$ that represents
the tree obtained from $\val(G)$ by marking each node
that is selected by the automaton $\mathcal{A}$. 
For each occurrence of a nonterminal $B$ in the right-hand
sides of the rules of $G$, there is at most one new nonterminal
of the form $(q,B,q_1,\dots,q_k)$, where   
$q,q_1,\dots,q_k$ are states of $\mathcal{A}$.
The construction is similar to the proof of Theorem~\ref{theorem:count}:
instead of computing $\varphi(A,q)=(q_1,\dots,q_k,n)$,
we construct a rule of the new grammar $G'$ of the form
$(q,A,q_1,\dots,q_k)\to t$, where $t$ is obtained from $P(A)$ by
replacing every nonterminal $B$ met in state $q'$ by the nonterminal
$(q',B,q_1',\dots, q_m')$ where 
$\varphi(B,q')=(q_1',\dots,q_k',n)$ for some $n$. When during such a run 
a selecting rule of $\mathcal{A}$ is applied to a terminal symbol $a$, then 
we relabel it by $\hat{a}$.
Finally, to be consistent with our definition of SLT grammars (which does
not allow non-reachable (useless) nonterminals because the hierarchical
order is required to be connected), we remove all non-reachable nonterminals
in one run through $G'$.

\begin{lemma}
\label{lemma1}
Let $G$ be a $k$-SLT grammar and $\mathcal{A}$ a DST automaton.
An SLT grammar $G'$ can be constructed in time 
$O(|Q||G|)$ so that $\val(G')$ is the
relabeling of $\val(G)$ according to $\mathcal{A}$.
\end{lemma}

Note that in Theorem~5 of~\cite{DBLP:journals/tcs/LohreyM06} it is shown
that membership of the tree $\val(G)$ with respect to a 
deterministic top-down tree automaton (dtta) can be checked in polynomial time.
The idea there is to construct a context-free grammar for the 
``label-paths'' of $\val(G)$; for a tree with root node $a$ and left child leaf $b$,
$a_1b$ is a label path. It then uses the property that the label-path language
of a dtta is effectively regular.

\textbf{Example.} Figure~\ref{fig:lemma1} shows an SLT grammar $G$
with $\val(G)=(aa)^8(e)$ and the DST automaton $\mathcal{A}$ for the XPath query
\[
Q_2=//\! *\![\text{count}(\text{ancestor::\! *\!})\ \text{mod}\ 3=2].
\]
While we do not translate queries using count and ancestor, 
the automaton for this particular query is easy to construct: it uses three states 
$q_1,q_2,q_3$ to count the number of nodes modulo three.
For simplicity the example is on a monadic tree, not an XML tree;
therefore the rules of $\mathcal{A}$ are of the form $q,\%\to q'$, i.e.,
the right-hand side contains only one state instead of two.
The figure also shows the SLT grammar $G'$, representing the relabeling according to
Lemma~\ref{lemma1}. One can verify that $G'$ produces the correct relabeled
tree, by computing $\val(G'):$ 
\[
\begin{array}{l}
(q_1,A_0)\rightarrow\langle q_1,A_1,q_3\rangle(\langle q_3,A_1,q_2\rangle(e))\rightarrow\\
\langle q_1,A_2,q_2\rangle(\langle q_2,A_2,q_3\rangle(\langle q_3,A_2,q_1\rangle(\langle q_1,A_2,q_2\rangle(e))))\rightarrow\\
\langle q_1,A_3,q_3\rangle(\langle q_3,A_3,q_2\rangle(\langle q_2,A_3,q_1\rangle(\langle q_1,A_3,q_3\rangle(\\
\langle q_3,A_3,q_2\rangle(\langle q_2,A_3,q_1\rangle(\langle q_1,A_3,q_3\rangle (\langle q_3,A_3,q_2\rangle(e)...)\rightarrow\\
a(a(\hat a(a(a(\hat a(a(a(\hat a(a(a(\hat a(a(a(\hat a(a(e)...)\\
\end{array}
\]

\begin{figure}
\centering
\begin{tabular}{lcl}
SLT grammar $G$: && DST automaton $\mathcal{A}$:\\
$\begin{array}{l}
A_0\rightarrow A_1(A_1(e))\\
A_1(y)\rightarrow A_2(A_2(y))\\
A_2(y)\rightarrow A_3(A_3(y))\\
A_3(y)\rightarrow a(a(y))
\end{array}$&$\quad$&
$\begin{array}{l}
q_1,\%\to q_2\\
q_2,\%\to q_3\\
q_3,\%\To q_1\\
\end{array}$
\end{tabular}

\begin{tabular}{l}
\\
relabeling SLT grammar $G'$:\\
$\begin{array}{lcl}
\langle q_1,A_0\rangle   & \to & \langle q_1,A_1,q_3\rangle(\langle q_3,A_1,q_2\rangle(e))\\
\langle q_1,A_1,q_3\rangle(y) & \to & \langle q_1,A_2,q_2\rangle(\langle q_2,A_2,q_3\rangle(y))\\
\langle q_1,A_2,q_2\rangle(y) & \to & \langle q_1,A_3,q_3\rangle(\langle q_3,A_3,q_2\rangle(y))\\
\langle q_1,A_3,q_3\rangle(y) & \to & a(a(y))\\
\langle q_3,A_3,q_2\rangle(y) & \to & \hat a(a(y))\\
\langle q_2,A_2,q_3\rangle(y) & \to & \langle q_2,A_3,q_1\rangle(\langle q_1,A_3,q_3\rangle(y))\\
\langle q_2,A_3,q_1\rangle(y) & \to & a(\hat a(y))\\
\langle q_3,A_1,q_2\rangle(y) & \to & \langle q_3,A_2,q_1\rangle(\langle q_1,A_2,q_2\rangle(y))\\
\langle q_3,A_2,q_1\rangle(y) & \to & \langle q_3,A_3,q_2\rangle(\langle q_2,A_3,q_1\rangle(y))\\
\end{array}$
\end{tabular}

\caption{\label{fig:lemma1} A relabeling SLT grammar $G'$ with start
  production $\langle q_1,A_0\rangle$, for a given SLT grammar $G$ with respect 
to a DST automaton $\mathcal{A}$ for query $Q_2$.}
\end{figure}

\begin{theorem}
\label{theo:mat}
Let $G$ be an SLT grammar
and $\mathcal{A}$ be a DST automaton.
Let $r=|\mathcal{A}(\val(G))|$.
We can compute an ordered list of pre-order numbers of the nodes
in $\mathcal{A}(\val(G))$ in time
$O(|Q||G| + r)$.
\end{theorem}
\begin{proof}
Let $G=(N,S,P)$.
By Lemma~\ref{lemma1} we obtain in time $O(|Q||G|)$ an SLT grammar $G'$
whose tree $\val(G')$ is the relabeling of $\val(G)$ with
respect to $\mathcal{A}$. The list of pre-order numbers is constructed
during two passes through the grammar $G'$. First we compute bottom-up 
for each nonterminal $A$ (of rank $k$)
the off-sets of all relabeled nodes that appear in $P(A)$.
An offset is a pair of integers $(c,o)$ where 
$0\leq c\leq k$ is a chunk number, and $o$ is the position of
a node within a chunk.
A \emph{chunk} is the part of the pre-order traversal of $P(A)$ that is
before, between, or after parameters. I.e. when $A$ is of rank $k$, then
there are $k+1$ chunks: the chunk of the traversal from the root of $P(A)$ to 
the first parameter $y_1$ which has
chunk number $0$, the chunks of the traversal between
two parameters $y_i$ and $y_{i+1}$ (with number $i$), and the chunk after the last parameter $y_k$
with number $k$.
We construct a mapping $\varphi$ that maps a nonterminal $A$, a state $q$, and a chunk number $c$
to a pair $(n,L)$ where $n$ is the total number of nodes in the chunk and 
$L$ is the list of off-sets, in order.
We now do a complete pre-order traversal through the grammar $G'$, while
maintaining the current-preorder number $u$ in a counter.
When we meet a nonterminal $A$ in chunk $c$ with a non-empty list $L$ of
off-sets, we add $u$ to each offset and append the resulting list to the output
list.
\end{proof}

\subsection{Serialization}

Here we want to output the XML serialization of the result subtrees rooted
at the result nodes of a query (given by a DST automaton). 
Again, we want the output in pre-order.

\begin{theorem}\label{theo:ser}
Let $G$ be an SLT grammar and $\mathcal{A}$ a DST automaton.
Let $s$ be the sum of sizes of all subtrees rooted at the nodes
in $\mathcal{A}(\val(G))$.
We can output all result subtrees of $\mathcal{A}(\val(G))$ in time 
$O(|Q||G| + s)$.
\end{theorem}

\begin{proof}
The proof is similar to the proof of Theorem~\ref{theo:mat} in that it runs in two passes over grammar $G'$ whose tree $\val(G')$ is the relabeled one according to Lemma~\ref{lemma1}. During the bottom-up run through the grammar, we construct a mapping $\varphi$ that maps a nonterminal $A$, a state $q$, and a chunk number $c$ to a sequence $S$ of opening and closing brackets of the pre-order traversal corresponding to $A$, $q$, and $c$. Then during the complete pre-order traversal though $G'$ we construct
 a sequence $S'$ of opening and closing brackets containing only result subtrees of $\mathcal{A}(\val(G))$ and pointers to marked elements for nested result nodes. At a nonterminal $A$, in a state $q$, and a chunk $c$ we first start appending to $S'$ if $\varphi(A,q,c)$ contains a marked node. Then when meeting nonterminals $A$, in state $q$, and chunk $c$ inside marked nodes subtrees we always append $\varphi(A,q,c)$ to $S'$, and we store pointers to marked nodes.
Finally, based on the obtained sequence $S'$, the selected subtrees are 
serialized by following the $|\mathcal{A}(\val(G))|$ pointers to their roots in $S'$.
\end{proof}

We can do better, if we are allowed to output a compressed
representation of the concatenation of all result subtrees.
In fact, the result stated in Theorem~\ref{theo:ser}, follows
from Theorem~\ref{theo:bad}.

We can construct a straight-line string grammar (SLP) in time 
$O(|G|)$ that generates the pre-order traversal of the tree
$\val(G)$, see Figure~\ref{fig:lemma2} for an example. 
But, what about an SLP that outputs the concatenation of all pre-order traversals of
the \emph{marked subtrees}?
What is the size of such a grammar?
If every node is marked, and the original tree has $N$
nodes, then the length of the represented string is in
$O(N^2)$.
\begin{figure}[htp]
\begin{tabular}{l}
$\ \ S\ \to\ \langle 113,0\rangle\langle 312,0\rangle\!<\!e\!>\!<\!/e\!>\!\langle 312,1\rangle\langle 113,1\rangle$\\
$\begin{array}{lcl}
\langle 113,0\rangle & \to & \langle 122,0\rangle\langle 223,0\rangle\\
\langle 312,0\rangle & \to & \langle 321,0\rangle\langle 122,0\rangle\\
\langle 312,1\rangle & \to & \langle 122,1\rangle\langle 321,1\rangle\\
\langle 113,1\rangle & \to & \langle 223,1\rangle\langle 122,1\rangle\\
\langle 122,0\rangle & \to & \langle 133,0\rangle\langle 332,0\rangle\\
\langle 122,1\rangle & \to & \langle 332,1\rangle\langle 133,1\rangle\\
\langle 223,0\rangle & \to & \langle 231,0\rangle\langle 133,0\rangle\\
\langle 223,1\rangle & \to & \langle 133,1\rangle\langle 231,1\rangle\\
\langle 321,0\rangle & \to & \langle 332,0\rangle\langle 231,0\rangle\\
\langle 321,1\rangle & \to & \langle 231,1\rangle\langle 332,1\rangle\\
\langle 133,0\rangle & \to & <\!a\!><\!a\!>\\
\langle 133,1\rangle & \to & <\!/a\!><\!/a\!>\\
\langle 332,0\rangle & \to & <\!\hat a\!><\!a\!>\\
\langle 332,1\rangle & \to & <\!/a\!><\!/\hat a\!>\\
\langle 231,0\rangle & \to & <\!a\!><\!\hat a\!>\\
\langle 231,1\rangle & \to & <\!/\hat a\!><\!/a\!>\\
\end{array}$
\end{tabular}
\caption{\label{fig:lemma2} SLP grammar $G''$ for the pre-order traversal 
of $\val(G')$
of Figure~\ref{fig:lemma1}, where $\langle ijk,l\rangle$ is a new nonterminal of $G''$ denoting the pair of a nonterminal $\langle q_i,A_j,q_k\rangle$ of $G'$ and $l$ the number of its chunk.}
\end{figure}

\begin{theorem}\label{theo:bad}
Given an SLT grammar $G$ 
and a subset $R$ of the nodes of $\val(G)$, an SLP $P$ 
for the concatenation of all subtrees
at nodes in $R$ (in pre-order) can be constructed in time
$O(|G||R|)$.
\end{theorem}
\begin{proof}
We assume that the nodes in $R$ are given as pre-order numbers.
Let us first observe that for a given SLT grammar $H$, an SLP
grammar of the pre-order traversal of $\val(H)$,
using opening and closing labeled brackets (for instance in XML syntax)
can be constructed in time and space $O(|H|)$, following
the proof of Theorem~3 of~\cite{DBLP:journals/is/BusattoLM08}
(they state $O(|G|k)$ because they count the number of nonterminals
of the SLP).
In one preprocessing pass through $G$ we compute the length
of every chunk of every nonterminal.
Let now $u$ be a pre-order number in $R$. 
Using the information of the chunk lengths, we can determine, starting
at the right-hand side of the start nonterminal, which nonterminal 
generates the node $u$. We keep the respective subtree of the right-hand side,
and continue building a larger sentential tree, until we obtain a sentential form
that has the desired terminal node of $u$ at its root. The obtained sentential
tree $t$ is of size $O(|G|)$. We introduce a new nonterminal $S_u$ with rule
$S_u\to t$. This process is repeated for each node in $R$. Finally we construct
a new start rule which in its right-hand side has the concatenation of all $S_u$'s
with $u\in R$. The size of the resulting grammar is $O(|G||R|)$. Finally, we produce
the SLP for the traversal strings, as mentioned above. 
\end{proof}

Let us consider milder tree compression via DAGs~\cite{DBLP:conf/vldb/KochBG03},
by $0$-SLT grammars that do not use parameters $y_j$.
In this case we can improve the result of
Theorem~\ref{theo:bad} as follows.

\begin{theorem}\label{theo:dag}
Given a $0$-SLT grammar $G$ and a subset $R$ of the nodes
of $\val(G)$, an SLP $G'$ for the concatenation of all subtrees
at nodes in $R$ (in pre-order) can be constructed such that $G'$ is of size
$O(|G|+|R|)$.
\end{theorem}
\begin{proof}
We first bring the grammar $G$ into ``node normal form''.
This means that the right-hand side of each rule contains
exactly one terminal symbol. 
Note that this may increase the number of nonterminals, 
but does not change the size of the grammar. 
Now, each subtree of $\val(G)$ is represented by a 
unique nonterminal.
The grammar $G'$ is obtained from $G$ by considering $G$ as
a string grammar in the obvious way, and then changing the 
start production such that its right-hand side is the
concatenation (in pre-order) 
of the nonterminals corresponding to nodes in $R$.
\end{proof}
It is easy to extend Theorem~\ref{theo:dag} to slightly more
general compression grammars: the hybrid DAGs of
Lohrey, Maneth, and Noeth~\cite{DBLP:conf/icdt/LohreyMN13}. A hybrid DAG of an unranked tree is
obtained by first building the minimal unranked DAG, then 
constructing its first-child next-sibling encoding (seen as a grammar),
and then building the minimal DAG of this grammar. 
The hybrid DAG of an unranked tree is guaranteed smaller (or equal to)
the minimal unranked DAG and the minimal binary DAG (= DAG of 
first-child next-sibling encoded binary trees).
Theorem~\ref{theo:dag} is extended by bringing the unranked DAG
into node normal form.

\pagebreak

\section{XPath Filters}\label{sect:filters}

An XPath filter (in our fragment)
checks for the existence of a path, starting
at the current node. It is written in the form
$[./p]$ where $p$ is an XPath query as before.
For instance, the query
\[
//b[.//c/d/e][./a/b]/f/g
\]
first selects those $b$-nodes that have somewhere
below the path c/d/e, and which also have an $a$-child that
has a $b$-child. Starting from such $b$-nodes, the query selects
the $f$-children, and then the $g$-children thereof.

It is well-known that such filters can be evaluated using
\emph{deterministic bottom-up tree automata}.
For each filter path $p$ in the query we build one bottom-up
automaton (this construction is very similar to our earlier
construction of DST automata), in time linear to the size of the $p$.
We then build the product automaton $\mathcal{A}$ of all the filter automata. 
The size of this automaton is the product of the sizes of all
filter paths in the query. If we run this automaton over a given
input tree, then it will tell us for each node of the tree,
which filter paths are true at that node.
Thus, for a given SLT grammar~$G$, if we build the intersection grammar
with our bottom-up filter automaton $\mathcal{A}$, then 
the new nonterminals (and terminals) are of the form
\[
(p, A, p_1,\dots, p_m)
\]
where $m$ is the rank of $A$ and $p,p_1,\dots,p_m$ are $n$-tuples of 
filter states. Such a tuple $p$ tells us the states of each filter automaton
and hence the truth value of all the filters.

Given an XPath query with filters, we first build the combined filter
automaton $\mathcal{A}$. We then build for a given SLT grammar $G$, the
bottom-up intersection grammar $G_{\mathcal{A}}$.
We remove the filters from the query and build the
DST automaton $\mathcal{B}$ as before.
However, now we annotate the rules of this automaton, by information
about filters: if at a step of the query that corresponds to state $q$ 
of the $\mathcal{B}$ the filters $f_1,\dots,f_m$ appear in the query,
then the $q$-rule is annotated by these filters; when we evaluate top-down 
we check whether the filters are true, using the annotated information
of the intersection grammar $G_{\mathcal{A}}$.
It is shown in Theorem~1 of~\cite{DBLP:journals/tcs/LohreyM06}
that for a bottom-up automaton and a $k$-SLT grammar, 
the intersection grammar can be produced in time $O(|Q|^{k+1}|G|)$. 

\begin{theorem}\label{theo:filter}
Let $G$ be an SLT grammar
and $\mathcal{A}$ a DST automaton with filter automata
$F_1,\dots,F_n$; the sets of states are $Q, Q_1,\dots, Q_n$, respectively.
Let $r=|\mathcal{A}(\val(G))|$ and $k$ be the rank of $G$. 
We can construct a grammar $G'$ which represents
$\val(G)$ with all result nodes marked,
in time $O(|Q|(|Q_1|\cdots|Q_n|)^{k+1}|G|)$.
\end{theorem}

The complexity stated in Theorem~\ref{theo:filter} is rather
pessimistic and we believe that it can be improved.
We are applying a result about deterministic bottom-up automata
from~\cite{DBLP:journals/tcs/LohreyM06}. We do want to execute
our filter automata bottom-up, but, they are indeed deterministic
top-down automata. In future research we would like to improve
the worst-case complexity stated in the theorem above by taking
this into account.
Consider filters over the child axis only, e.g., $[./a/b/c]$.
Instead of using a bottom-up automaton for the filter and 
constructing an intersection grammar according to~\cite{DBLP:journals/tcs/LohreyM06}
in time $O(|Q|^{k+1}|G|)$, we use a top-down automaton for the
``relative'' query $./a/b/c$; it can be constructed similar as
our DST automata. Via Lemma~\ref{lemma1} we obtain a marking grammar $G'$
in time $O(|Q||G|)$. We now want to transform this grammar so that instead
of the $c$-nodes, their grandparent $a$-nodes are marked.
How expensive is this transformation?
It seems n the worst case that each occurrence of a nonterminal in $G'$ must
be changed into a distinct copy (and recursively for the new right-hand
sides). This would run in time $O(|G'|^2)$. Can it be improved?
How can be handle other axes such as descendant? In which cases is
this solution more efficient than the one of Theorem~\ref{theo:filter}?

\bibliographystyle{eptcs}
\bibliography{lib}
\end{document}